\documentclass[runningheads]{llncs}
\usepackage[T1]{fontenc}

\usepackage{cite}

\usepackage{amsmath,amssymb,amsfonts}
\usepackage{graphicx}
\usepackage{textcomp}
\usepackage{xcolor}
\usepackage{balance}
\usepackage{url}  %Required
\usepackage{graphicx}  %Required
\usepackage{xcolor}
\usepackage{CJKutf8}
\usepackage[online]{optional}
\usepackage{url}
\usepackage{amsmath,amsfonts,amsthm,amssymb}
\usepackage{algorithm}
\usepackage[noend]{algpseudocode}
\usepackage{graphicx}
\usepackage{textcomp}
\usepackage{xcolor}
\usepackage{comment}
\usepackage{graphicx}
\usepackage[utf8]{inputenc}
\usepackage{CJK}
\usepackage{subfigure}
\usepackage{algorithm}
\usepackage{algpseudocode}
\usepackage{epsfig}
\usepackage{color}
\def\BibTeX{{\rm B\kern-.05em{\sc i\kern-.025em b}\kern-.08em
    T\kern-.1667em\lower.7ex\hbox{E}\kern-.125emX}}

\newcommand{\shenr}[1] % 'R'evisied by shen
{
\textbf{\color{red}{#1}}
}

\begin{document}
\begin{CJK}{UTF8}{bkai}

%\title{Identifying Onsite and Remote Working Groups for Minimizing the Risk of Infection in Pandemics}
%\titlerunning{Identifying Onsite and Remote Working Groups}

\title{Risk-Aware Skill-Coverage Hybrid Workforce Configuration on Social Networks}
\titlerunning{Risk-Aware Skill-Coverage Hybrid Workforce Configuration}

\author{Anonymous}
\authorrunning{Anonymous}

\author{
Hui-Ju Hung\inst{1} \and
Guang-Siang Lee\inst{2} \and
Chia-Hsun Lu\inst{3}\and
Chih-Ya Shen\inst{3}\and
De-Nian Yang\inst{2} 
}

\authorrunning{Hung et al.}

\institute{
National Central University, Taoyuan, Taiwan\\
\email{hjhung@ncu.edu.tw}
\and
Academia Sinica, Taipei, Taiwan\\
\email{gslee9822802@citi.sinica.edu.tw, dnyang@iis.sinica.edu.tw}
\and
National Tsing Hua University, Hsinchu, Taiwan\\ 
\email{chlu@m109.nthu.edu.tw, chihya@cs.nthu.edu.tw}
}

\begin{comment}
\author{Chia-Hsun Lu\inst{1} \and
Hui-Ju Hung\inst{2} \and
Guang-Siang Lee\inst{3} \and
Chih-Ya Shen\inst{2} \and
De-Nian Yang\inst{4} \and
}
%
\authorrunning{Lu et al.}

\institute{Department of Computer Science, National Tsing Hua University, Hsinchu, Taiwan.\\ 
\email{chlu@m109.nthu.edu.tw,}
\email{chihya@cs.nthu.edu.tw.}
\and
Department of Computer Science and Information Engineering, National Central University, Taoyuan, Taiwan.\\
\email{hjhung@ncu.edu.tw.}
\and
Research Center for Information Technology Innovation, Academia Sinica, Taipei, Taiwan.\\
\and
Institute of Information Science, Academia Sinica, Taipei, Taiwan.}
\end{comment}

\maketitle

% Given required skills and an upper bound on risks, the goal is to select a subset of employees to work onsite that maximizes a collaboration-based objective while satisfying skill and risk constraints.

\begin{abstract}
In hybrid workforce configurations, it is important to decide which employees should work onsite or remotely while ensuring the collaboration benefits against contact-based health risks and skill requirements. In this paper, we formulate the \textit{Risk-aware Skill-coverage Hybrid Workforce Configuration} (RSHWC) problem on a two-layer social network that balances physical contact risks and social collaboration ties to meet skill requirements.  We prove that RSHWC is NP-hard and propose the \textit{Guided Risk-aware Iterative Assembling} (\texttt{GRIA}) algorithm, a multi-stage algorithm that combines risk-aware workforce construction, skill-preserving workforce refinement, and risk-reducing member replacement. Experiments on four real-world networks show that \texttt{GRIA} consistently outperforms state-of-the-art baselines under various settings.
\end{abstract}

%\keywords{Graph query \and Onsite and remote employee arrangement \and Pandemics.}

\section{Introduction}

Hybrid work, combining on-site and remote arrangements, is now widespread across various sectors. Many prominent employers, including Meta, Microsoft, Citigroup, and Standard Chartered, have adopted long-term hybrid work schedules. In these companies, most staff spend two or three days in the office and work remotely on the remaining days. The reasons go beyond health concerns: hybrid work helps reduce office space costs and support employee well-being. Empirical studies show that such arrangements can improve satisfaction and retention, reduce commuting time, and enhance work-life balance, without sacrificing performance~\cite{bloom2022hybrid,bloom2024hybrid}. %aksoy2023time
In some settings, well-designed hybrid policies reduce attrition by about 30\%~\cite{bloom2022hybrid,bloom2024hybrid}. %At the same time, seasonal influenza and other respiratory infections still cause workplace outbreaks even outside pandemic periods. Allowing symptomatic employees to work from home remains crucial for maintaining operations while minimizing contact-based health risks.

Deciding who should work onsite in a given period requires sophisticated mining and planning. Organizations must cover location-specific tasks (e.g., front-desk reception, laboratory experiments), assemble onsite workers with the right skill combinations, and maintain operational continuity~\cite{barrero2023evolution,melkonian2023optimization}. Companies may also impose bounds on onsite days per employee or constraints derived from indicators of employee well-being~\cite{bloom2024hybrid,melkonian2023optimization}. However, physical co-location simultaneously exposes employees to contact-based risks, such as seasonal infectious diseases or other in-person operational risks subject to internal or regulatory constraints~\cite{aloisi2022essential}. Therefore, hybrid workforce configuration requires sophisticated mining and planning that couples skills, collaboration structures, and contact-based risks.

Existing research typically addresses only one side of this hybrid workforce configuration problem (i.e., selecting employees to work onsite). Commercial staff-scheduling methods allocate shifts and balance the workload, but largely treat employees as interchangeable and ignore their previous collaborative social structure~\cite{abadi2021hssaga,guerriero2022modeling,melkonian2023optimization}. Team formation in social networks assembles experts with complementary skills and cohesive connectivity~\cite{shen20222,Lappa2009,shen2016,shen2022,nikolaou2024team}, but do not account for hybrid presence or contact-based risks. Epidemic-control intervention in contact social networks selects nodes to intervene (e.g., vaccinate or quarantine) to limit disease propagation~\cite{epic1,epic7}, without modeling the productive value of onsite work or detailed skill requirements. In contrast, hybrid workforce configurations require smart algorithms that carefully examine social collaborations, skills, and contact-based risks.

%To capture the structure of hybrid work, we model the organization as a two-layer social network $G = (G_c, G_p)$ defined on a common set of employees $V$. The contact network $G_c = (V, E_c)$ represents potential physical interactions, and the partnership network $G_p = (V, E_p)$ represents collaboration ties and skills. Each employee $v \in V$ is associated with a skill set, and each partnership edge $(u,v) \in E_p$ has two collaboration scores: one for when both employees are onsite and one for when at least one of them works remotely. This two-layer model allows us to analyze how any onsite roster simultaneously affects contact-based propagation on $G_c$ and collaboration on $G_p$.

On top of this two-layer social network, we formulate the \emph{Risk-aware Skill-coverage Hybrid Workforce Configuration (RSHWC)} problem with a contact social network $G_c$ (describing the risk of contact interactions) and a collaboration social network $G_p$ (describing the expertise and social collaboration atmosphere). Given a required skill set $R$ and a propagation model for contact-based risks on $G_c$, the goal of RSHWC is to select a subset of onsite employees that maximizes a collaboration objective subject to two constraints: (i) \emph{skill coverage}, in which the onsite employees cover all required skills, and (ii) \emph{group-level propagation}, in which the contact-based risk on $G_c$ does not exceed a specified limit. The contact-based risks may come from epidemic diseases such as seasonal influenza or COVID-19, incorrect rumors propagated through word-of-mouth, or other contact-based risks relevant to hybrid work.
%The RSHWC problem is combinatorial and involves tightly coupled constraints, making it computationally challenging. 
We first prove that RSHWC is NP-hard and then propose \emph{Guided Risk-aware Iterative Assembling} (\texttt{GRIA}), a scalable algorithm that jointly accounts for skill coverage, collaborative gains from onsite presence, and upper bounds on contact-based risks, by risk-aware workforce construction, skill-preserving workforce refinement, and risk-reducing member replacement. 

Our main contributions are as follows:
\begin{itemize}
\item We explore hybrid workforce configuration considering social collaborations, skills, and contact-based risks.
\item We formulate RSHWC and prove that RSHWC is NP-hard.
\item We propose \texttt{GRIA}, a scalable algorithm with theoretical guarantees on the skill-preserving workforce refinement module.
\item We conduct extensive experiments on four real-world networks, showing that \texttt{GRIA} outperforms state-of-the-art methods. 
\end{itemize}

The rest of this paper is organized as follows.
Section~\ref{sec:relatedwork} reviews the relevant literature. Section~\ref{sec:problem} formulates the \emph{RSHWC} problem. Section~\ref{sec:algo} presents \texttt{GRIA}. Section~\ref{sec:exp} evaluates its performance, and Section~\ref{sec:conclu} concludes the paper.

\section{Related Work}
\label{sec:relatedwork}

% \noindent \textbf{Onsite and remote arrangement.} 
% Hybrid work arrangements have been extensively studied~\cite{batista2024personnel,guerriero2022modeling,moosavi2022staff,abadi2021hssaga}. For example, Guerriero and Guido employ linear programming to minimize onsite headcount~\cite{guerriero2022modeling}. Other approaches explicitly incorporate disease-related constraints, such as restricting cross-room assignments~\cite{moosavi2022staff}, balancing workloads~\cite{abadi2021hssaga}, or jointly optimizing testing and staffing decisions~\cite{batista2024personnel}. Although effective in their respective settings, these methods are typically computationally intensive and difficult to scale, making them unsuitable for our problem setting.

\noindent \textbf{Team formation in social networks.} 
The team formation problem focuses on constructing a group of experts that collectively satisfy task requirements~\cite{nikolaou2024team,shen20222,shen2016,Lappa2009}. 
Lappa et al.~\cite{Lappa2009} aim to form an expert team by minimizing the diameter of the team’s spanning tree, while Shen et al.~\cite{shen2016} instead seek teams that minimize pairwise distances among selected nodes. 
Chang et al.~\cite{shen20222} leverage reinforcement learning with diameter- and skill-based embeddings to identify expert teams in social networks. 
Recently, Nikolaou et al.~\cite{nikolaou2024team} formulate team formation amid conflicts, balancing individual task preferences against pairwise conflicts. 
However, the above schemes do not consider the contact-based risks and only extract homogeneous teams.

\noindent \textbf{Dense subgraphs.} 
Research on dense subgraphs seeks to identify cohesive vertex sets or communities~\cite{shen2022,zhou2024counting,Cui2014}. 
Existing work includes discovering subgraphs that maximize average edge weight under hop constraints~\cite{shen2022}, identifying $k$-clique-based dense structures~\cite{zhou2024counting}, and using local search for community detection~\cite{Cui2014}. %xu2024efficient
These works focus on connectivity and density but do not consider the influence of nodes and edges, which are crucial to model the contact-based risks.

\noindent \textbf{Prevention of epidemics and misinformation.} 
Optimization-based methods for controlling the spread of epidemics and misinformation have been widely investigated~\cite{muppasani2024expressive,epic1,epic7}. 
Typical strategies include selecting key nodes for intervention to reduce overall infection levels~\cite{epic1}, planning sequences of actions on these nodes~\cite{muppasani2024expressive}, and designing containment operations to limit propagation~\cite{epic7}. 
However, these approaches focus on spread minimization in isolation and do not integrate hybrid workforce configurations. 
%(e.g., staffing levels, coverage, and safety requirements), which are central to our problem.

\section{The RSHWC Problem}
\label{sec:problem}

In this section, we formulate the \emph{Risk-aware Skill-coverage Hybrid Workforce Configuration  (RSHWC)} problem on a two-layer social network and prove that it is NP-hard. 
%We first introduce the underlying model and then define the optimization problem.

\noindent\textbf{Preliminaries.}
We model the workforce as a two-layer social network
$G = (G_c, G_p)$ on a common vertex set $V$ of employees.
The \emph{contact network} $G_c = (V, E_c)$ captures potential physical interactions among the employees. Each contact edge $e \in E_c$ has an influence probability $\sigma(e) \in [0,1]$ that represents the probability that a terminal is influenced by the other terminal through $e$~\cite{contact_nw}. The contact-based risk on $G_c$ can be modeled by~\cite{kempe2003maximizing,SIR} and the influence probabilities can be learned by~\cite{goyal2010learning,zhang2019learning}.
%to estimate the contact-based risk on $G_c$ started from a set of initially infected employees $S \subseteq V$. 
For any vertex subset $U \subseteq V$, we denote $risk_U(S)$ as the expected number of infected vertices in $U$ when the set of initially infected persons is $S \cap U$. Without loss of generality, in the following we use epidemic terminology, and it can be applied to other risks, such as misinformation dissemination.

The \emph{partnership network} $G_p = (V, E_p)$ represents social collaboration relationships. Each employee $v \in V$ has a skill set $\mathcal{S}(v) \subseteq \mathcal{U}$ from
a global skill universe $\mathcal{U}$. Each partnership edge $[u,v] \in E_p$ carries an
onsite collaboration score $o([u,v]) \ge 0$ and a remote collaboration score
$r([u,v]) \ge 0$, quantifying the benefit of collaboration when both endpoints
are onsite versus when at least one works remotely. Fig.~\ref{ill_ex}
illustrates a toy example of this two-layer structure.

\begin{figure}[t]
    \centering
    \subfigure[L1: \emph{Contact network}]{\label{show_gc}
        \includegraphics[width=0.34\columnwidth]{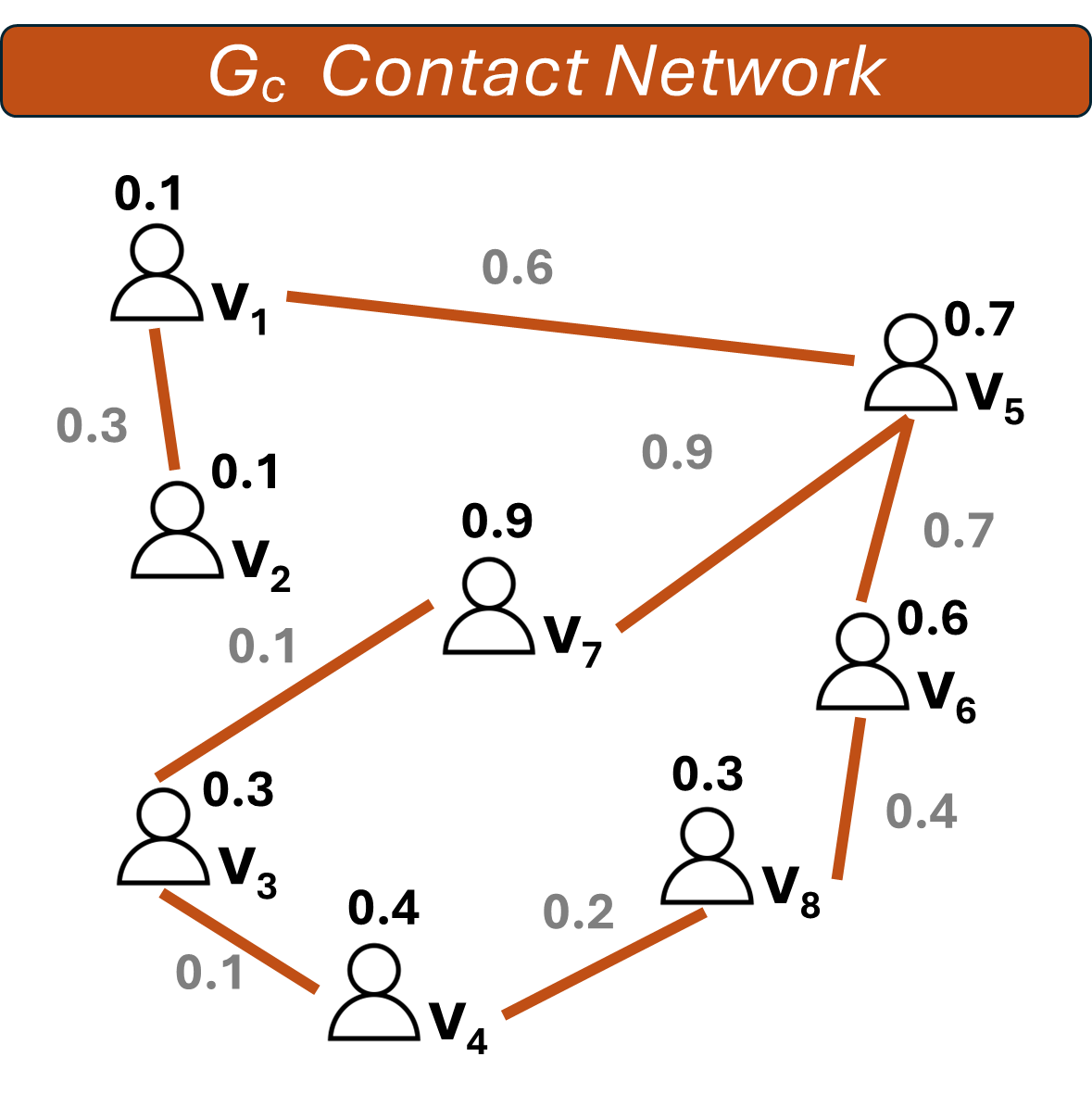}
    }
    %\hspace{10pt}
    \subfigure[L2: \emph{Partnership network}]{\label{show_ge}
        \includegraphics[width=0.35\columnwidth]{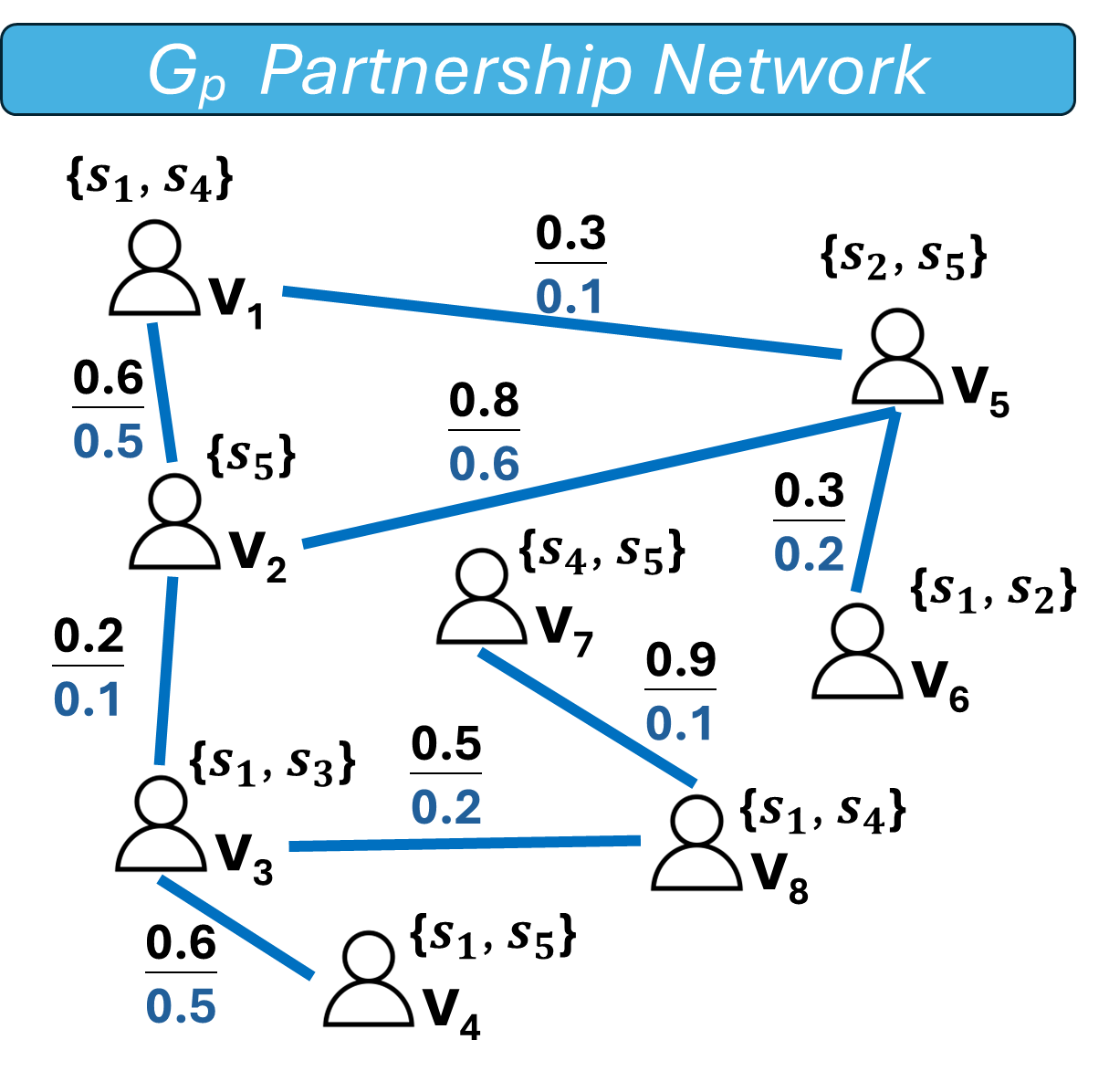}
    }
    \caption{Illustrative example of a two-layer social network}
    \label{ill_ex}
\end{figure}
%\hj{等佳勳改圖}
%\noindent\textbf{Problem formulation.}

%\emph{Collaboration objective.}
The \emph{collaboration objective} refers to the benefits of collaboration. In hybrid settings, each employee either works onsite or remotely for
a given period. For an onsite set $V^* \subseteq V$, the remaining employees
$V \setminus V^*$ work remotely.
On the partnership network $G_p$, each edge $[u,v] \in E_p$ contributes
$o([u,v])$ if both endpoints are onsite and $r([u,v])$ otherwise. Given an
onsite set of employees $V^*$, we define our objective as the \textit{average collaboration}:
\begin{align}
\small
\alpha(V^*)
    &=
    \frac{1}{|V^*|} \Big(
    \sum_{\substack{u,v \in V^* \\ [u,v] \in E_p}} o([u,v])
    +    
    \sum_{\substack{u \in V^*,\, v \notin V^* \\ [u,v] \in E_p}} r([u,v])
    +
    \sum_{\substack{u,v \notin V^* \\ [u,v] \in E_p}} r([u,v])\Big).
\label{eq:alpha}
\end{align}
% \begin{align}
% \small
% \alpha(V^*)
%     &=
%     \frac{1}{|V^*|}
%     \sum_{\substack{u,v \in V^* \\ [u,v] \in E_p}} o([u,v])
%     +
%     \frac{1}{|V^*|}
%     \sum_{\substack{u \in V^*,\, v \notin V^* \\ [u,v] \in E_p}} r([u,v])
%     +
%     \frac{1}{|V^*|} 
%     \sum_{\substack{u,v \notin V^* \\ [u,v] \in E_p}} r([u,v]).
% \label{eq:alpha}
% \end{align}
The first term aggregates onsite-onsite collaborations, while the second and
third terms aggregate collaborations with at least one remote endpoint. The
normalization by $|V^*|$ interprets $\alpha(V^*)$ as collaboration per onsite
employee.

Following~\cite{pastor2015epidemic,walters2018modelling,kempe2003maximizing,SIR,hung2016routing}, the \textit{contact-based risk} is defined as the expected number of employees to be influenced/infected when the seed set is $S$, and the influence only propagates on the induced subgraph of the contact network with only onsite employees. The expected number of influenced employees can be derived according to models like SIR, SEIR, SIRS, IC, LT~\cite{pastor2015epidemic,kempe2003maximizing,SIR}.

% For a given onsite set $V^* \subseteq V$, we run the SIR process on the induced
% contact graph $G_c[V^*]$ with seed set $S \cap V^*$. Let $\mathrm{Inf}(V^*;S)$ denote
% the number of infected onsite employees in this process, and let
% $\mathbb{E}[\mathrm{Inf}(V^*)]$ be its expectation.

\begin{definition}[The RSHWC Problem]
\label{def:RSHWC}
Let $G = (G_c, G_p)$ be a two-layer network as above. Each employee
$v \in V$ has a skill set $\mathcal{S}(v) \subseteq \mathcal{U}$, each contact edge $e \in E_c$ has
influence probability $\sigma(e)$, and each partnership edge $[u,v] \in E_p$
has onsite and remote collaboration scores $o([u,v])$ and $r([u,v])$,
respectively. Given a required skill set $\mathcal{R} \subseteq \mathcal{U}$, an initial infected set $S \subseteq V$, and an epidemic budget $C > 0$, the
\emph{Risk-aware Skill-coverage Hybrid Workforce Configuration  (RSHWC)} problem is to find an onsite set $V^* \subseteq V$ that maximizes $\alpha(V^*)$ subject to
(i) \textit{skill coverage}: $ \bigcup_{v \in V^*} \mathcal{S}(v) \supseteq \mathcal{R}$ and (2) \textit{group-level epidemic control}: $risk_{V^*}(S) \leq C$.
\end{definition}

RSHWC couples combinatorial social collaboration, skill coverage, and non-linear propagation, in which every employee is either \textit{onsite} or \textit{remote}. The next theorem proves even the decision version of RSHWC is NP-hard.

%The RSHWC problem is challenging because it couples combinatorial skill coverage, heterogeneous collaboration gains under onsite and remote modes, and non-linear propagation constraints on the contact network. We next show that even the decision version of RSHWC is NP-hard.

\begin{theorem}\label{RSHWC:NPhard}
The RSHWC problem is NP-hard.
\end{theorem}

\begin{proof}
We prove NP-hardness via a reduction from the classical Exact Cover by 3‑Sets (\emph{X3C}) problem~\cite{X3C}.
An instance of X3C consists of a ground set $X = \{x_1, x_2, \dots, x_{3q}\}$ and a collection $\mathcal{C} = \{C_1, C_2, \dots, C_n\}$ of 3-element subsets $C_i \subseteq X$.
The question is to determine whether there exists a subcollection $\mathcal{C}_x \subseteq \mathcal{C}$ of size $|\mathcal{C}_x| = q$ such that every element of $X$ appears in exactly one set in $\mathcal{C}_x$; in this case $\mathcal{C}_x$ is called an \emph{exact cover} of $X$.

Given an X3C instance $(X,\mathcal{C})$, we construct an instance of RSHWC as follows. (1) We create one vertex $v_i$ for each subset $C_i \in \mathcal{C}$, so $V = \{v_1,\dots,v_n\}$. The skill universe is $\mathcal{U} = X$, and the required skill set is $\mathcal{R} = X$. For each vertex $v_i$, we define its skill set as $\mathcal{S}(v_i) = C_i$; that is, each skill corresponds to an element of $X$. (2)We let $G_p = (V, E_p)$ be the complete graph on $V$. For every edge $[u,v] \in E_p$, we set $o([u,v]) = 1$ and $r([u,v]) = 1$. Thus, all collaboration scores are identical, and the objective $\alpha(V^*)$ depends only on the size of $V^*$, not on which particular vertices are chosen.
(3) We let $G_c = (V, E_c)$ be the complete graph on $V$.
%For every vertex $v \in V$, we set its initial infection probability to $\iota(v) = 1$.
Consider that $S=V$, i.e., all nodes are infected initially.
For every edge $e \in E_c$, we set the propagation probability to $\sigma(e) = 0$.
% Under this parameterization, there is no secondary transmission: infections do not spread along edges, and the expected number of infected individuals in a subset $V^*$ under the SIR model equals
% \[
%     \mathbb{E}[\mathrm{Inf}(V^*)] = \sum_{v \in V^*} \iota(v) = |V^*|.
% \]
%We set the infection threshold to $p = 1$, so the individual infection constraint $\iota(v) \le p$ is always satisfied.
Also, we set the epidemic limit to $C = q$.
Suppose the X3C instance admits an exact cover $\mathcal{C}_x = \{C_{i_1},\dots,C_{i_q}\}$. Consider the subset $V^* = \{v_{i_1},\dots,v_{i_q}\}$.
By construction, $\bigcup_{v \in V^*} \mathcal{S}(v) = \bigcup_{j=1}^q C_{i_j} = X = \mathcal{R}$, so the skill coverage constraint is satisfied.
%Each selected vertex has $\iota(v) = 1 \le p$, so the individual infection constraint holds.
%$\mathbb{E}[\mathrm{Inf}(V^*)] = |V^*| = q = C$
Moreover, $risk_{V^*}(S) = |V^*| = q = C$, so the epidemic limit is also satisfied. Thus, $V^*$ is a feasible solution to the constructed RSHWC instance.

Conversely, suppose there exists a feasible onsite set $V^* \subseteq V$ for the constructed RSHWC instance. Since each vertex covers at most three skills, and $\mathcal{R}$ contains $3q$ skills, the skill coverage constraint implies $3\,|V^*| \;\ge\; \bigl|\bigcup_{v \in V^*} \mathcal{S}(v)\bigr| \;\ge\; |\mathcal{R}| = 3q$, so $|V^*| \ge q$.
On the other hand, by our choice of parameters, and the group-level epidemic control constraint $risk_{V^*}(S) =|V^*| \le C = q$ yields $|V^*| \le q$. Hence $|V^*| = q$.
Furthermore, we have
$\bigl|\bigcup_{v \in V^*} \mathcal{S}(v)\bigr| = |\mathcal{R}| = 3q$, 
and $V^*$ contains exactly $q$ vertices, each with $|\mathcal{S}(v)| = 3$.
If two selected skill sets $\mathcal{S}(v_i)$ and $\mathcal{S}(v_j)$ overlapped on some element $x \in X$, the size of the union would be strictly smaller than $3q$.
Therefore, the sets $\{\mathcal{S}(v): v \in V^*\}$ are pairwise disjoint and together cover $X$, so they form an exact cover of $X$.
Mapping each vertex $v_i \in V^*$ back to the corresponding subset $C_i$ yields an exact cover for the original X3C instance.
We have thus given a polynomial-time reduction from X3C to RSHWC. Since X3C is NP-complete, even the decision version of RSHWC is NP-hard. The theorem follows.
\end{proof}

\begin{algorithm}[t]
  \scriptsize
  \caption{Guided Risk-aware Iterative Assembling (\texttt{GRIA})}
  \label{alg:GRIA}
  \begin{algorithmic}[1]
    \Require A two-layer network $G = (G_c, G_p)$, skill set $\mathcal{S}(v)$, $\forall v \in V$, influence probability $\sigma(e)$, $\forall e \in E_c$, onsite collaboration scores $o(e')$ and remote collaboration scores $r(e')$, $\forall e' \in E_p$, required skills set $\mathcal{R}$, an initial infected set $S \in V$, and an epidemic budget $C$ 
    \Ensure A subset $V^* \subseteq V$ 
    \Statex \#Risk-aware workforce construction (RWC)
    \State $V^* \gets \emptyset$, $V_{cand} \gets V$, and $\mathcal{R}_{miss} \gets \mathcal{R}$
    \While{$\mathcal{R}_{miss} \neq \varnothing$}
        \State $V_{feas} \gets \{v \in V_{cand} | \mathcal{S}(v) \cap \mathcal{R}_{miss} \neq \varnothing \land risk_{V^*}(S) \leq C\}$
        \State $v_{best} \gets \arg\max_{v \in V_{feas}} \frac{\tau_{V^*}(v)}{\mathrm{risk}_{V^*}(\{v\})}$
        \State $V^* \gets V^* \cup \{v_{best}\}$
        \State $V_{cand} \gets V_{cand} \setminus \{v_{best}\}$, $\mathcal{R}_{miss} \gets \mathcal{R}_{miss} \setminus \mathcal{S}(v_{best})$
    \EndWhile
    \While{$risk_{V^*}(S) \leq C$}
        \State $v_{best} \gets \arg\max_{v \in V_{cand}} \frac{\tau_{V^*}(v)}{\mathrm{risk}(v)}$
        \State $V^* \gets V^* \cup \{v_{best}\}$, and $V_{cand} \gets V_{cand} \setminus \{v_{best}\}$
    \EndWhile
    \Statex \#Skill-preserving workforce refinement (SWR)
    \For{each vertex $v \in V^*$}
        \If{$\tau_{V^*}(v) < \alpha(V^*)$ and $\bigcup_{u\in V^*\setminus\{v\}} \mathcal{S}(u) \supseteq \mathcal{R}$}
            \State $V^* \gets V^* \setminus \{v\}$
        \EndIf
    \EndFor
    \Statex \#Risk-reducing member replacement (RMR)
    \State $change \gets \textbf{true}$ 
    \While{$change$}
        \State $change \gets \textbf{false}$ 
        
        \State $v_o \gets \arg\max_{v \in V^*} \mathrm{risk_{V^*}}(\{v\})$
        \State $v_r \gets \arg\min_{v \in V\setminus V^*} \mathrm{risk_{V^*\cup\{v\}}}(\{v\})$
        \State $V' \gets V^*\setminus\{v_o\} \cup \{v_r\}$
        \If{$\bigcup_{u\in V'} \mathcal{S}(u) \supseteq \mathcal{R}$ and $\alpha(V') \geq \alpha(V^*)$ and $risk_{V'}(S) \leq C$}
            \State $V^* \gets V'$, $change \gets \textbf{true}$
        \EndIf
    \EndWhile\\
    \Return $V^*$
\end{algorithmic}
\end{algorithm}

\section{The Proposed \texttt{GRIA} Algorithm}
\label{sec:algo}

%Guided Risk-aware Iterative Assembling (GRIA) constructs an onsite set$V^* \subseteq V$ in three phases: Initial onsite team construction, local exchange for risk reduction, and skill-preserving workforce refinement. \hj{這裡描述一些concept level of each component}
\textit{Guided Risk-aware Iterative Assembling} (\texttt{GRIA}) constructs an onsite set $V^* \subseteq V$ in three conceptually distinct phases. The first phase performs a risk-aware workforce construction that incrementally builds an onsite team covering all required skills while favoring employees that bring large collaboration gains per unit of contact risk.  The second phase performs skill-preserving workforce refinement, removing onsite employees whose marginal contribution to collaboration is below the current average, which yields a while onsite team while preserving feasibility.
The final phase involves risk-reducing member replacement to replace high-risk onsite employees with safer alternatives without sacrificing skill coverage or collaboration. The pseudo cod of \texttt{GRIA} is presented in Algorithm~\ref{alg:GRIA}.

\noindent\textbf{Risk-aware workforce construction (RWC).}
\texttt{GRIA} starts from an empty onsite set $V^* = \emptyset$ and a set of candidate employees $V_{\text{cand}} = V$ to be selected to work onsite. Let $\mathcal{R}_{\text{miss}}$ denote the set of skills in $\mathcal{R}$ that are not yet covered by $V^*$. \texttt{GRIA} repeatedly chooses one employee from $V_{\text{cand}}$ and adds it to $V^*$.
While $\mathcal{R}_{\text{miss}}$ is non-empty, only employees who help cover at least one missing skill are considered, i.e., vertices $v$ with $\mathcal{S}(v) \cap \mathcal{R}_{\text{miss}} \neq \emptyset$ while the infection risk is bounded by $C$. Among the candidates who can cover more than one skill yet covered, \texttt{GRIA} selects the employee maximizing
$\frac{\tau_{V^*}(v)}{\mathrm{risk}(v)}$, favoring candidates that bring large additional onsite collaboration with the current team per unit of contact risk. 
After adding $v$ to $V^*$, \texttt{GRIA} removes the newly covered skills and the selected employee from $\mathcal{R}_{\text{miss}}$ and $V_{\text{cand}}$, respectively.
After all required skills are covered (i.e., $\mathcal{R}_{\text{miss}} = \emptyset$), \texttt{GRIA} selects the employee maximizing
$\frac{\tau_{V^*}(v)}{\mathrm{risk}(v)}$
continue to add employees as long as the propagation budget allows, and the collaboration benefits are positive. For any vertex subset $U \subseteq V$ and vertex $v \in V$, we define the \emph{collaboration gain} of $v$ with respect to $U$ as
$\tau_U(v) = \sum_{u \in U \cap N_p(v)} \bigl(o([v,u]) - r([v,u])\bigr)$,  where $N_p(v)$ is the neighbor set of $v$ in the partnership network. Intuitively, $\tau_U(v)$ measures the increase in collaboration obtained by bringing $v$ onsite when vertices in $U$ are potential onsite partners.
On the contact network, we define the risk function $risk_{V^*}(S)$ as the number of infected people when the infection is propagated through the contact network $G_c$, which can be quickly computed by existing methods~\cite{zhou2015upper,tang2015influence}.
Risk-aware workforce construction terminates when no remaining employee can be added
without violating the group-level epidemic constraint or bringing positive collaboration gain.

\noindent\textbf{Skill-preserving workforce refinement (SWR).}
After an initial onsite team is built, $V^*$ removes onsite employees whose collaboration contribution is strictly below the current average. Let $\alpha(V^*)$ be the average collaboration score of the current onsite set. For each $v \in V^*$, \texttt{GRIA} computes its marginal gain $\tau_{V^*}(v)$, checks whether $ \tau_{V^*}(v) < \alpha(V^*)$ and removing $v$ preserves skill coverage, $\bigcup_{u \in V^* \setminus \{v\}} \mathcal{S}(u) \supseteq \mathcal{R}$.
Whenever both conditions hold, $v$ is pruned from $V^*$. Because removing vertices can only reduce the infection risk. The correctness of the theoretical proof is detailed in Section~\ref{subsec:pruning}.
%, which shows that such skill-preserving workforce refinement never decreases the average collaboration score of any superset of the remaining vertices. 

\noindent\textbf{Risk-reducing member replacement (RMR).}
In the initial onsite team construction phase, \texttt{GRIA} may still place high-risk employees onsite if they offer
large initial gains or have critical skills. To reduce the infection phase, \texttt{GRIA} performs a local exchange phase between onsite and offsite employees. More specifically, in each iteration, \texttt{GRIA} identifies
(i) the onsite employee $v_o$ with the largest $\mathrm{risk_{V^*}}(\{v\})$ and
(ii) the remote employee $v_r=V \setminus V^*$ in $V_{\text{rem}}$ with the smallest
$\mathrm{risk_{V^*}}(\{v\})$ and try to exchange $v_o$ and $v_r$.  The exchange is accepted only if the solution $V' = V^* \setminus \{v_o\} \cup \{v_r\}$ satisfies the following three conditions: (i) the skill coverage is preserved,
$ \bigcup_{u \in V'} \mathcal{S}(u) \supseteq \mathcal{R}$;  (ii) the collaboration objective does not decrease $ \alpha(V') \ge \alpha(V^*)$,
and (iii) the group epidemic constraint remains satisfied,
$ risk_{V'}(S) \le C$. 
\texttt{GRIA} terminates after $t_e$ trials, where $t_e$ is a given parameter.

\noindent\textbf{Complexity analysis.}  
Initial onsite team construction runs for at most $O(|V|)$ iterations because at most $|V|$ vertices can be added into $|V^*|$. In each iteration, \texttt{GRIA} scans all remaining candidates and calculates the infection risk to choose the node to be added into $|V^*|$, requiring $O(|V|\cdot c_r)$ time, where $c_r$ is the running time for calculating the infection risk. Thus, this stage has complexity $O(|V|^2 \cdot c_r)$. Then, the skill-preserving workforce refinement scans all vertices in $V^*$ and performs $O(1)$ work per vertex, yielding an $O(|V|)$ cost. Since we only remove the nodes, the infection risk only reduces, and we do not need to recalculate it. The local exchange for risk reduction repeats $t_e$ times, and each time \texttt{GRIA} examines $|V|$ vertices to identify $v_o$ and $v_r$ to be exchanged, leading to $O(t_e \cdot |V| \cdot c_r)$ time. Overall, the total complexity of \texttt{GRIA} is $O(|V|^2\cdot c_r + |V| + t_e|V| \cdot c_r) = O(|V|^2 \cdot c_r)$, where $c_r$ is $O(|E|)$ according to~\cite{zhou2015upper}.

\subsection{Correctness Proof of Skill-preserving Workforce Refinement}
\label{subsec:pruning}

We now prove the correctness of the skill-preserving workforce refinement in \texttt{GRIA}.
Given a vertex subset $V^*$, let $c$ denote its average collaboration score, i.e., $\alpha(V^*) = c$.
Let $C' = \{u \in V^* \mid \tau_{V^*}(u) < c\}$ be the set of vertices whose collaboration gain is strictly below the current average.

\begin{lemma}\label{neighbor_prune1}
Let $C' = \{u \mid \tau_{V^*}(u) < c,\, u \in V^*\}$ and $Q \subset V^*$ with $\alpha(Q) \ge c$.
For every subset $\overline{C} \subset C'$, we have $\alpha(Q \cup \overline{C}) < \alpha(Q)$. \opt{short}{(Proof in~\cite{reprod})}
\end{lemma}

\opt{online}{
\begin{proof}
We prove the lemma by induction on $|\overline{C}|$.

\noindent\textbf{Base case.}
Consider $\overline{C} = \{u\}$ with $u \in C'$.
Then
\[
\small
    \alpha(Q \cup \{u\})
    = \frac{\alpha(Q)\cdot|Q| + \tau_{Q \cup \{u\}}(u)}{|Q| + 1}.
\]
Since $\tau_{Q \cup \{u\}}(u) \le \tau_{V^*}(u) < c$, we obtain
\[
\small
    \alpha(Q \cup \{u\})
    < \frac{\alpha(Q)\cdot|Q| + c}{|Q| + 1}.
\]
Because $\alpha(Q) \ge c$, it follows that
\[
\small
    \alpha(Q \cup \{u\}) - \alpha(Q)
    < \frac{\alpha(Q)\cdot|Q| + c}{|Q|+1}
      - \frac{\alpha(Q)\cdot(|Q|+1)}{|Q|+1}
    = \frac{c - \alpha(Q)}{|Q|+1} \le 0,
\]
so $\alpha(Q \cup \{u\}) < \alpha(Q)$.

\noindent\textbf{Inductive step.}
Assume the claim holds for any subset of $C'$ with size $k$.
Let $\overline{C}_k \subset C'$ with $|\overline{C}_k| = k$ and consider $\overline{C}_{k+1} = \overline{C}_k \cup \{u\}$ with $u \in C' \setminus \overline{C}_k$.
By the induction hypothesis, $\alpha(Q \cup \overline{C}_k) < \alpha(Q)$.
Then
\[
\small
    \alpha(Q \cup \overline{C}_{k+1})
    = \frac{\alpha(Q \cup \overline{C}_k)\cdot(|Q| + k)
      + \tau_{Q \cup \overline{C}_{k+1}}(u)}{|Q| + k + 1}
    < \frac{\alpha(Q)\cdot(|Q| + k) + c}{|Q| + k + 1},
\]
and hence
\[
\small
    \alpha(Q \cup \overline{C}_{k+1}) - \alpha(Q)
    < \frac{\alpha(Q)\cdot(|Q| + k) + c}{|Q| + k + 1}
      - \frac{\alpha(Q)\cdot(|Q| + k + 1)}{|Q| + k + 1}
    = \frac{c - \alpha(Q)}{|Q| + k + 1} \le 0.
\]
Thus $\alpha(Q \cup \overline{C}_{k+1}) < \alpha(Q)$, completing the induction.
\end{proof}
}

The next lemma shows that adding vertices in $C'$ cannot raise the average above $c$ when starting from a subset with average below $c$.

\begin{lemma}\label{neighbor_prune2}
Let $C' = \{u \mid \tau_{V^*}(u) < c,\, u \in V^*\}$ and $Q \subseteq V^*$ with $\alpha(Q) < c$.
For any subset $\overline{C} \subseteq C'$, we have $\alpha(Q \cup \overline{C}) < c$. \opt{short}{(Proof in~\cite{reprod})}
\end{lemma}

\opt{online}{
\begin{proof}
Again we use induction on $|\overline{C}|$.

\noindent\textbf{Base case.}
For $\overline{C} = \{u\}$ with $u \in C'$, we have
\[
\small
    \alpha(Q \cup \{u\})
    = \frac{\alpha(Q)\cdot|Q| + \tau_{Q \cup \{u\}}(u)}{|Q| + 1}
    < \frac{\alpha(Q)\cdot|Q| + c}{|Q| + 1}.
\]
Since $\alpha(Q) < c$, it follows that
\[
\small
    \alpha(Q \cup \{u\})
    < \frac{c\cdot|Q| + c}{|Q| + 1} = c.
\]

\noindent\textbf{Inductive step.}
Assume the claim holds for any $k$ vertices added from $C'$.
Let $\overline{C}_k \subset C'$ with $|\overline{C}_k| = k$, and let $\overline{C}_{k+1} = \overline{C}_k \cup \{u\}$ with $u \in C' \setminus \overline{C}_k$.
Using $\alpha(Q \cup \overline{C}_k) < c$ by induction, we obtain
\[
\small
    \alpha(Q \cup \overline{C}_{k+1})
    = \frac{\alpha(Q \cup \overline{C}_k)\cdot(|Q| + k)
      + \tau_{Q \cup \overline{C}_{k+1}}(u)}{|Q| + k + 1}
    < \frac{c\cdot(|Q| + k) + c}{|Q| + k + 1} = c.
\]
Thus $\alpha(Q \cup \overline{C}) < c$ for any $\overline{C} \subseteq C'$, completing the proof.
\end{proof}
}

\begin{theorem}[Skill-preserving workforce refinement]\label{IVP}
Given a solution $V^*$ with average collaboration score $\alpha(V^*) = c$, any vertex $u \in V^*$ such that $\tau_{V^*}(u) < c$ can be safely removed without decreasing the average collaboration score of any superset of the remaining vertices. \opt{short}{(Proof in~\cite{reprod})}
\end{theorem}

\opt{online}{
\begin{proof}
Let $V^*$ be the current solution with $\alpha(V^*) = c$ and $u \in C'$.
By Lemma~\ref{neighbor_prune1}, adding $u$ (or any subset of $C'$) to any subset $Q$ with $\alpha(Q) \ge c$ strictly decreases the average.
By Lemma~\ref{neighbor_prune2}, adding vertices from $C'$ to any $Q$ with $\alpha(Q) < c$ never raises the average to $c$ or above.
Thus, for any subset $Q \subseteq V^* \setminus \{u\}$, reintroducing $u$ cannot increase the average collaboration score of $Q$.
Hence $u$ can be safely pruned from $V^*$.
\end{proof}
}

\section{Experiments}
\label{sec:exp}

In this section, we evaluate the effectiveness and efficiency of the proposed algorithm \texttt{GRIA} and compare it with several baselines on multiple real-world networks. In the following, we first detail the experiment settings. Then, we show the experiment results.

\noindent\textbf{Datasets.}
We adopt four publicly available datasets: \emph{Manhattan}~\cite{MHA}, \emph{Virginia}~\cite{science2022}, \emph{ca-GrQc}~\cite{data_ca}, and \emph{ca-HepPh}~\cite{data_ca}. Table~\ref{tab:dataset} summarizes their statistics. For the datasets missing some information, following~\cite{teng2025breeding,hsu2024social,teng2022epidemic}, we augment the datasets with~\cite{teng2025breeding,teng2022epidemic,chen2024after,teng2024multi,hung2020efficient}.
%teng2021influence
% (i) \emph{Manhattan} is derived from a contact-tracing study in a rural US college town, containing 55,489 individuals and 216,683 interactions~\cite{MHA}.  
% (ii) \emph{Virginia} is constructed to assess medical costs during a large-scale epidemic in the state of Virginia, and includes 223,787 individuals and 918,523 interactions~\cite{science2022}.  
% (iii) \emph{ca-GrQc} represents coauthorship relations among authors who submitted papers to the General Relativity and Quantum Cosmology category on arXiv, with 5,242 authors and 14,496 coauthor edges~\cite{data_ca}.  
% (iv) \emph{ca-HepPh} is a coauthorship network in the High Energy Physics category on arXiv, with 12,008 authors and 118,521 edges~\cite{data_ca}.  
%Table~\ref{tab:dataset} summarizes the basic statistics of these datasets.

\begin{table}[t]
\centering
\small
\caption{Summary of datasets}
\label{tab:dataset}
%\resizebox{0.6\columnwidth}{!}{
\begin{tabular}{|c|c|c|c|c|}
\hline
Dataset        &  $|V|$ & $|E|$ & Ave. Deg. & Ave. Clustering Coefficient\\
\hline
Manhattan      &  55{,}489  & 216{,}683   & 7.80  & 0.32\\
Virginia       &  223{,}787 & 918{,}523   & 8.20  & 0.33\\
ca-GrQc        &  5{,}242   & 14{,}496    & 5.53  & 0.53 \\
ca-HepPh       &  12{,}008  & 118{,}521   & 19.74 & 0.61 \\
\hline
\end{tabular}
%}
\end{table}

\noindent\textbf{Baselines.}
We compare \texttt{GRIA} with six state-of-the-art baselines: \texttt{ST-Exa}~\cite{ST-Exa}, \texttt{cCoreExact}~\cite{cCoreExact}, \texttt{MaxGF}~\cite{shen2022}, \texttt{RF}~\cite{Lappa2009}, \texttt{RWR}~\cite{base_R}, and \texttt{EpRec}~\cite{base_E}.  
These baselines cover dense-subgraph/community search, team formation, and epidemic-aware recommendation, providing a diverse set of competitive methods.

\noindent\textbf{Metrics.}
Our evaluation is based on two metrics: (1) the \emph{objective value}, which is the average collaboration score $\alpha(V^*)$ of the selected onsite set, and (2) the \emph{computation time}, which is the running time required to compute the solution.

%\noindent\textbf{Experiment settings.}
%AS all datasets provide only a single-layer network, we construct the second layer by duplicating the original network and randomly perturbing its edges. Skill sets and edge influence probabilities are generated following previous studies~\cite{shen20222,zhou2015upper}. For each edge $e$, $o(e)$ is obtained via Jaccard similarity on neighbors and $r(e)$ is set to $r(e)=0.9\cdot o(e)$ following ~\cite{yang2022effects}. The $C$ is set to 30\% of employees.

\begin{figure}[htbp]
    \centering
	\subfigure[\emph{Manhattan}]{\label{MHA_skill_obj}
	    \includegraphics[width=0.35\columnwidth]{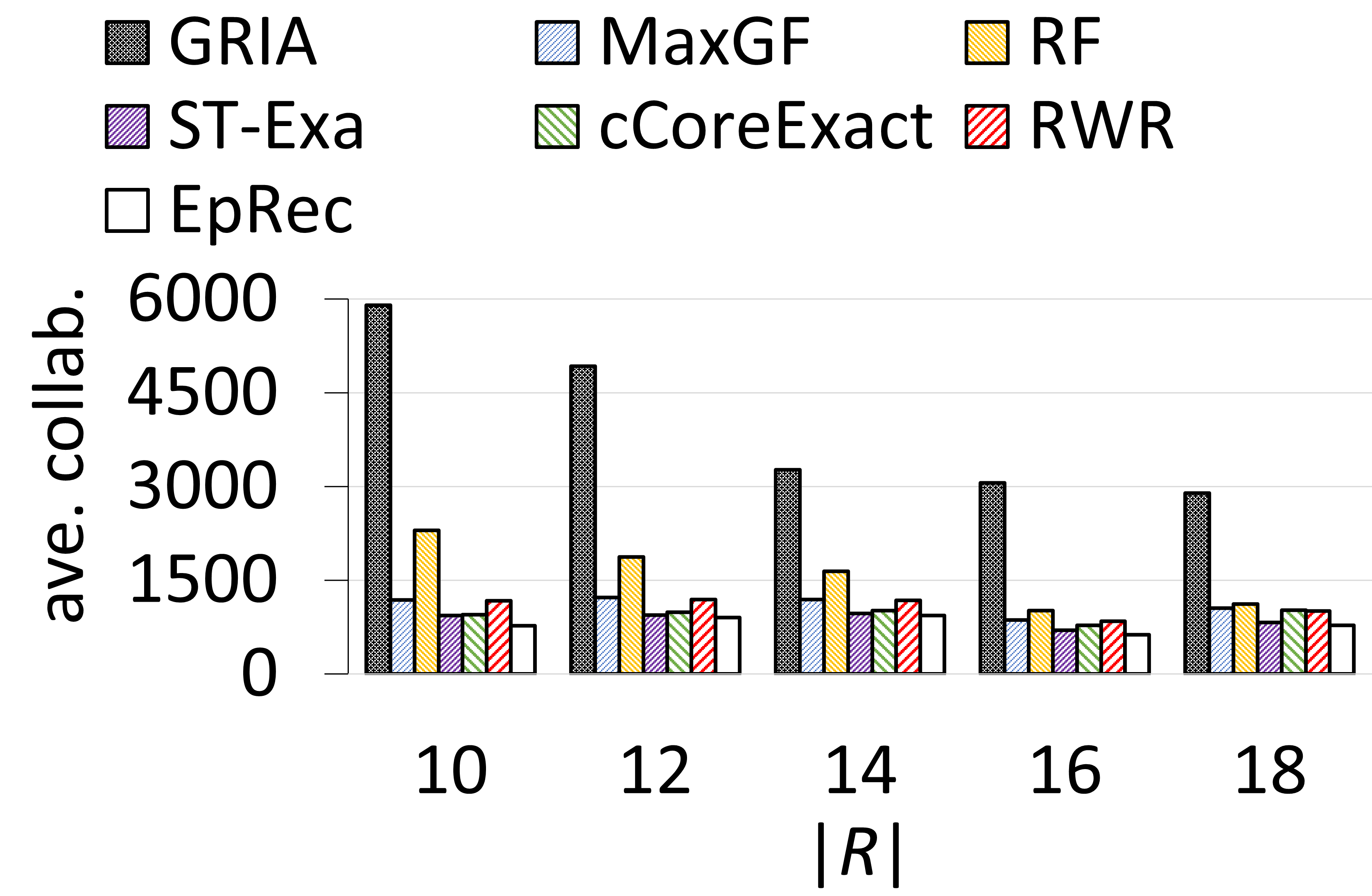}%
		}%
    \centering
	\subfigure[\emph{Virginia}]{\label{VA_skill_obj}
	    \includegraphics[width=0.35\columnwidth]{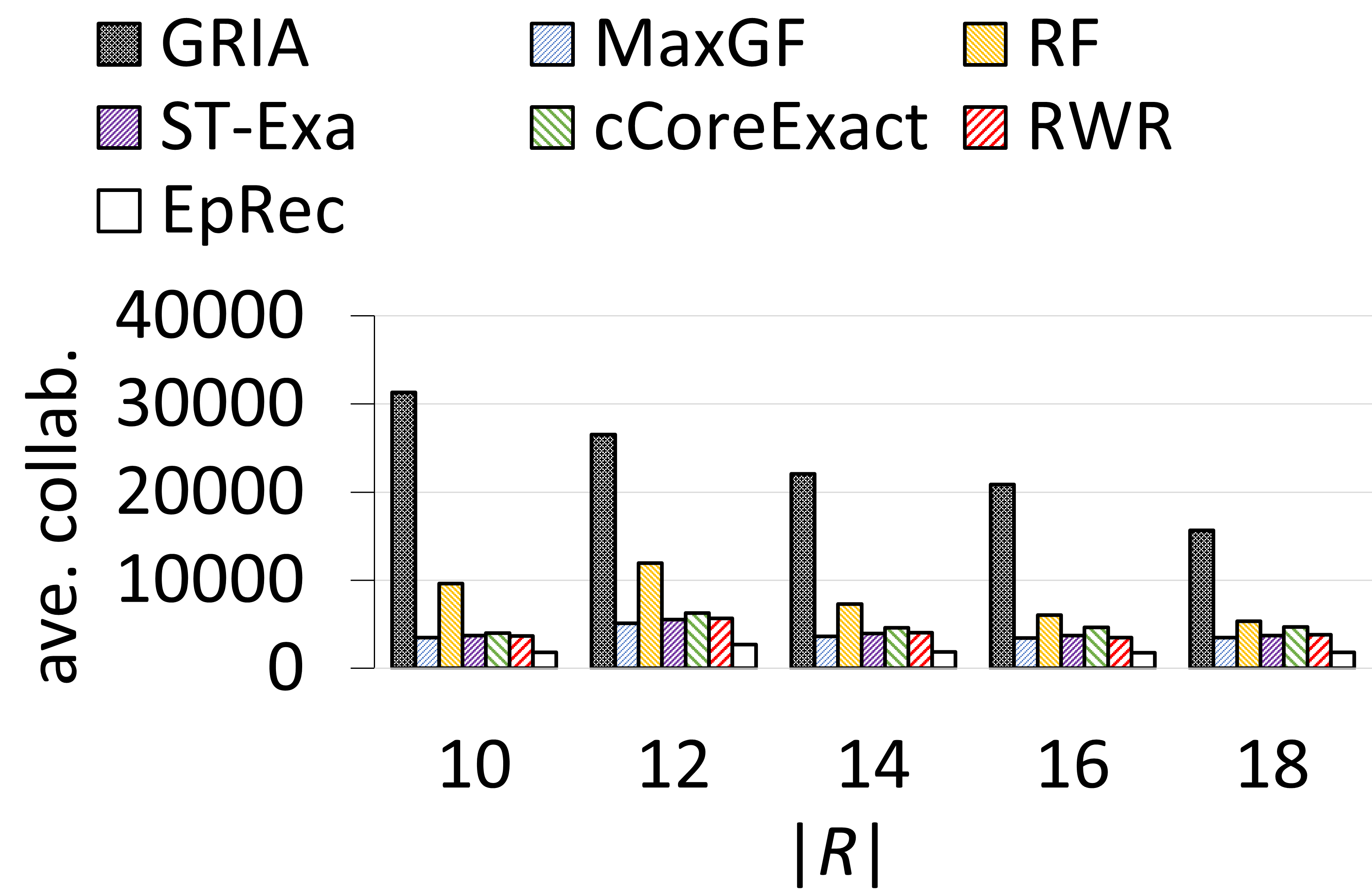}
		}
    \centering
	\subfigure[\emph{ca-GrQc}]{\label{GR_skill_obj}
	    \includegraphics[width=0.35\columnwidth]{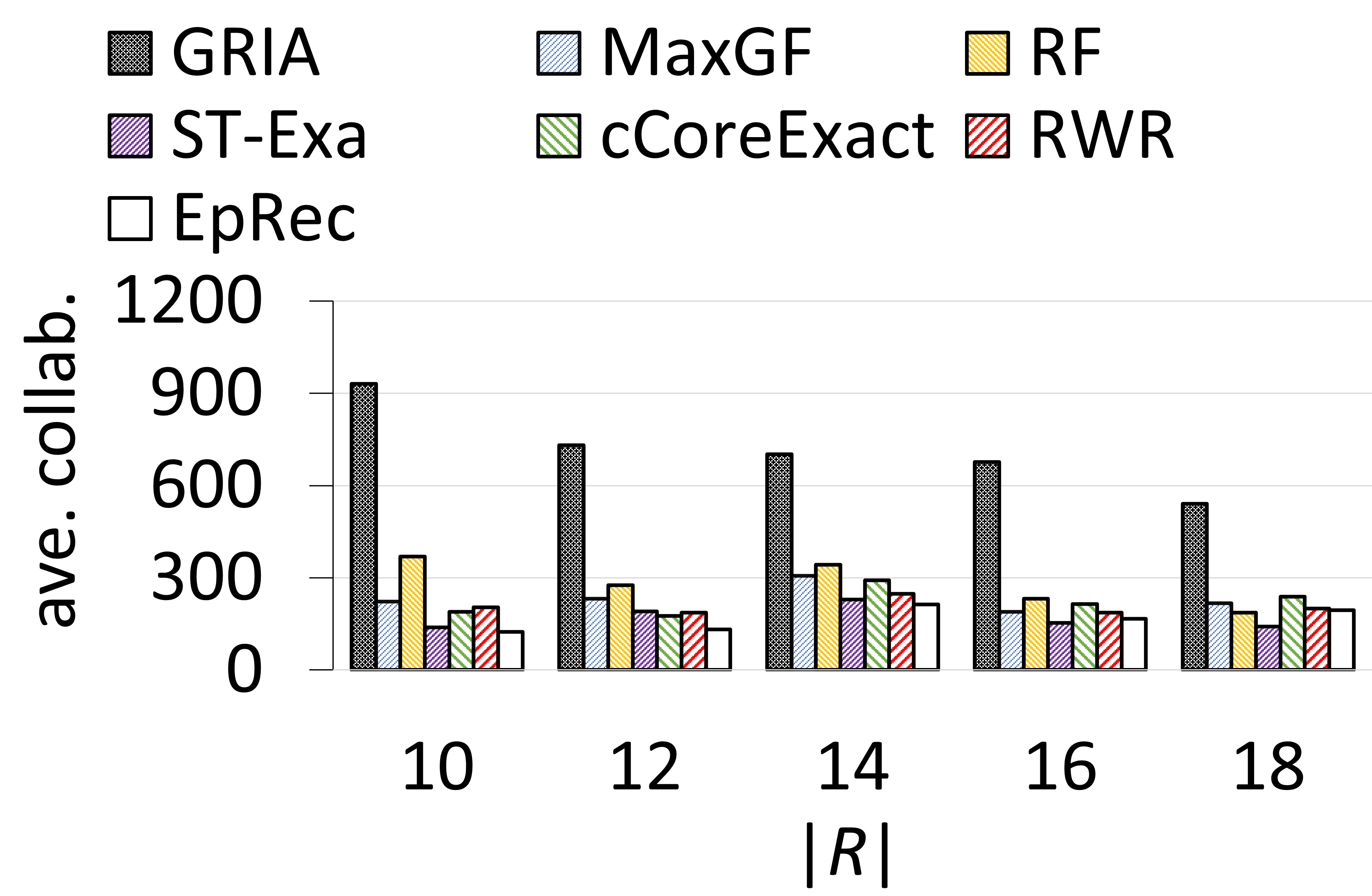}
		}%
    \centering
	\subfigure[\emph{ca-HepPh}]{\label{HE_skill_obj}
	    \includegraphics[width=0.35\columnwidth]{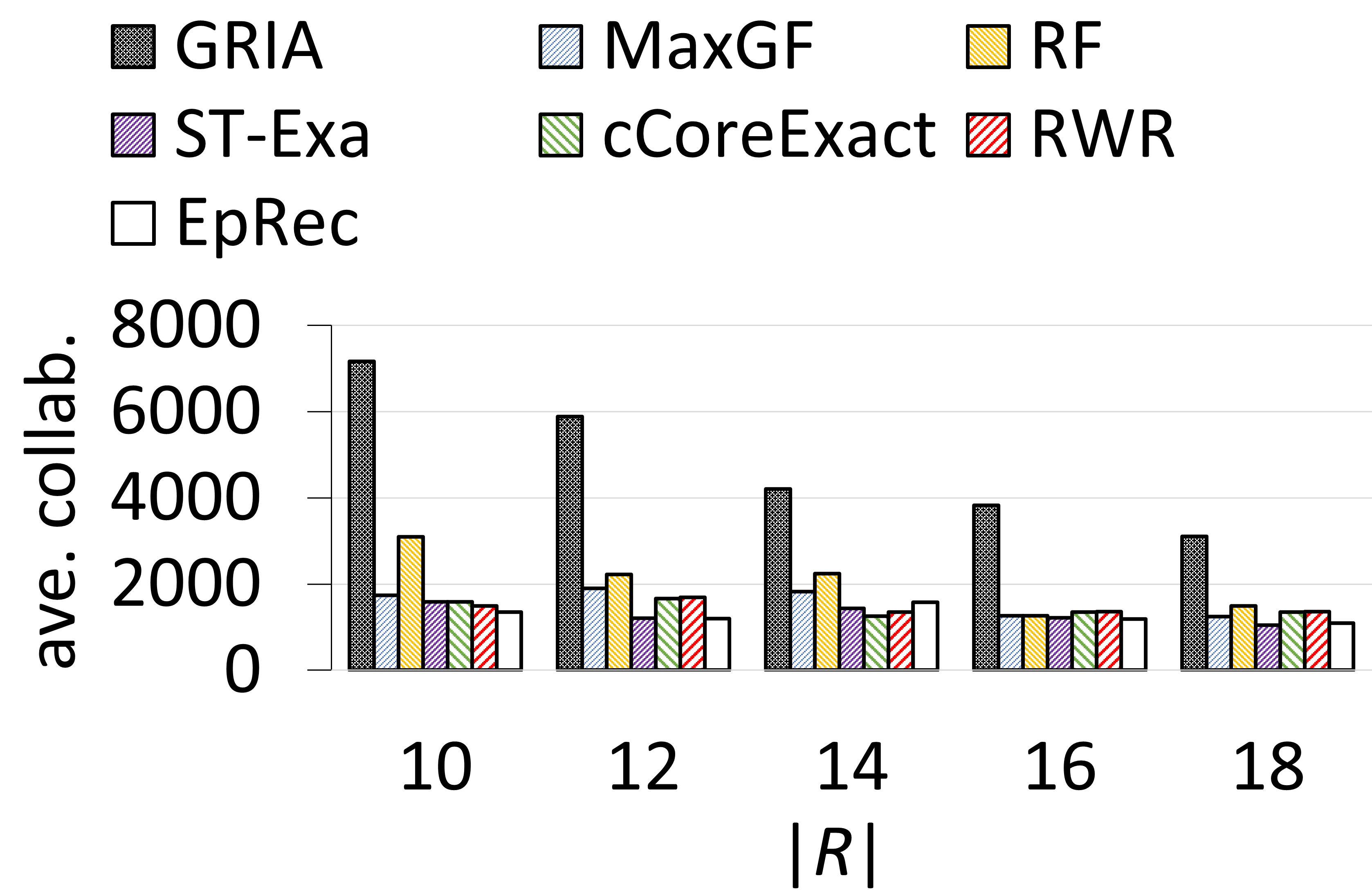}
		}%
    \caption{The average collaboration (ave. collab.) while varying $|R|$}
    \label{skill_obj}
\end{figure}

\begin{figure}[t]
    \centering
	\subfigure[\emph{Manhattan}]{\label{MHA_time}
	    \includegraphics[width=0.35\columnwidth]{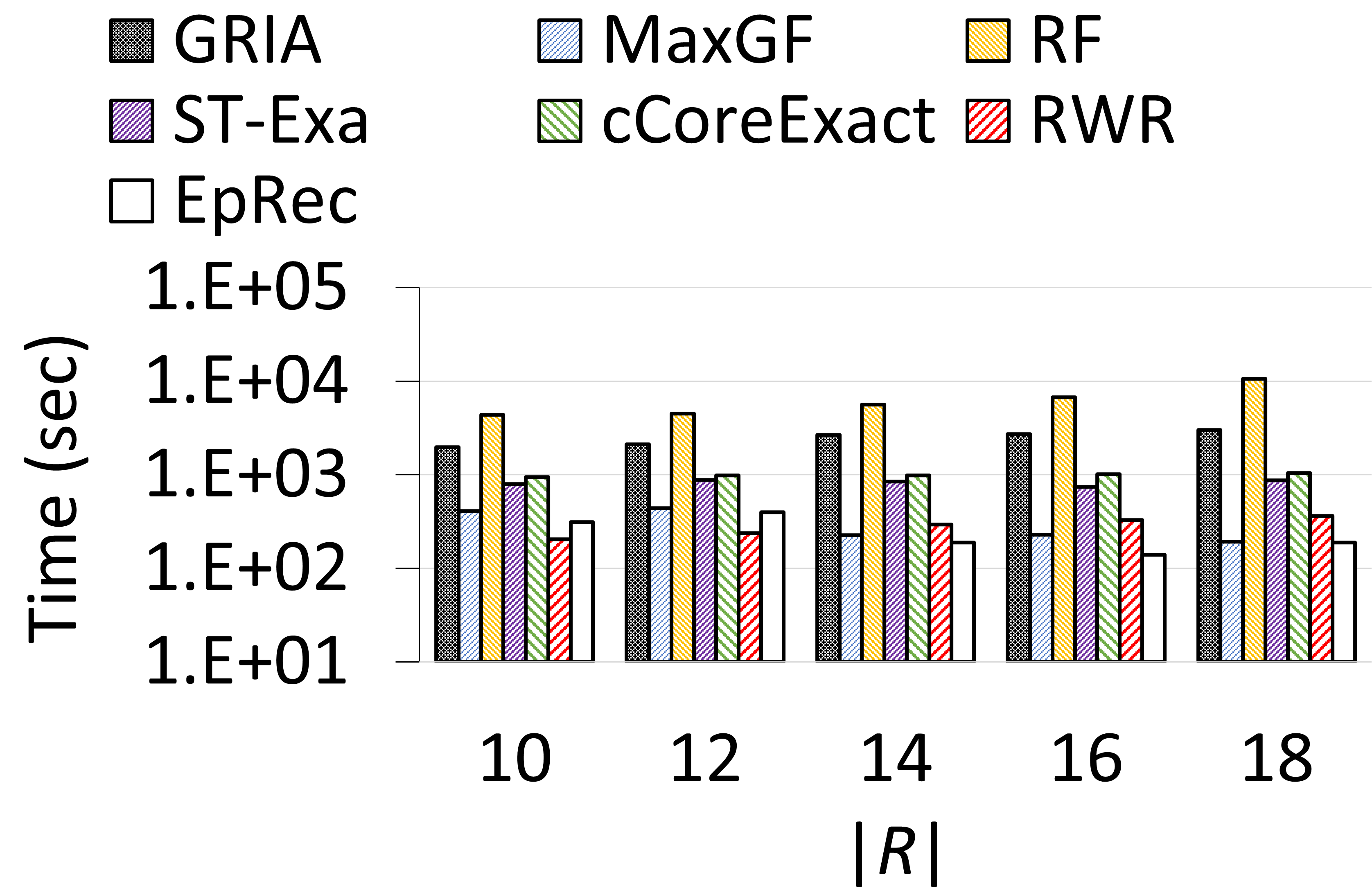}%
		}%
    \centering
	\subfigure[\emph{Virginia}]{\label{VA_time}
	    \includegraphics[width=0.35\columnwidth]{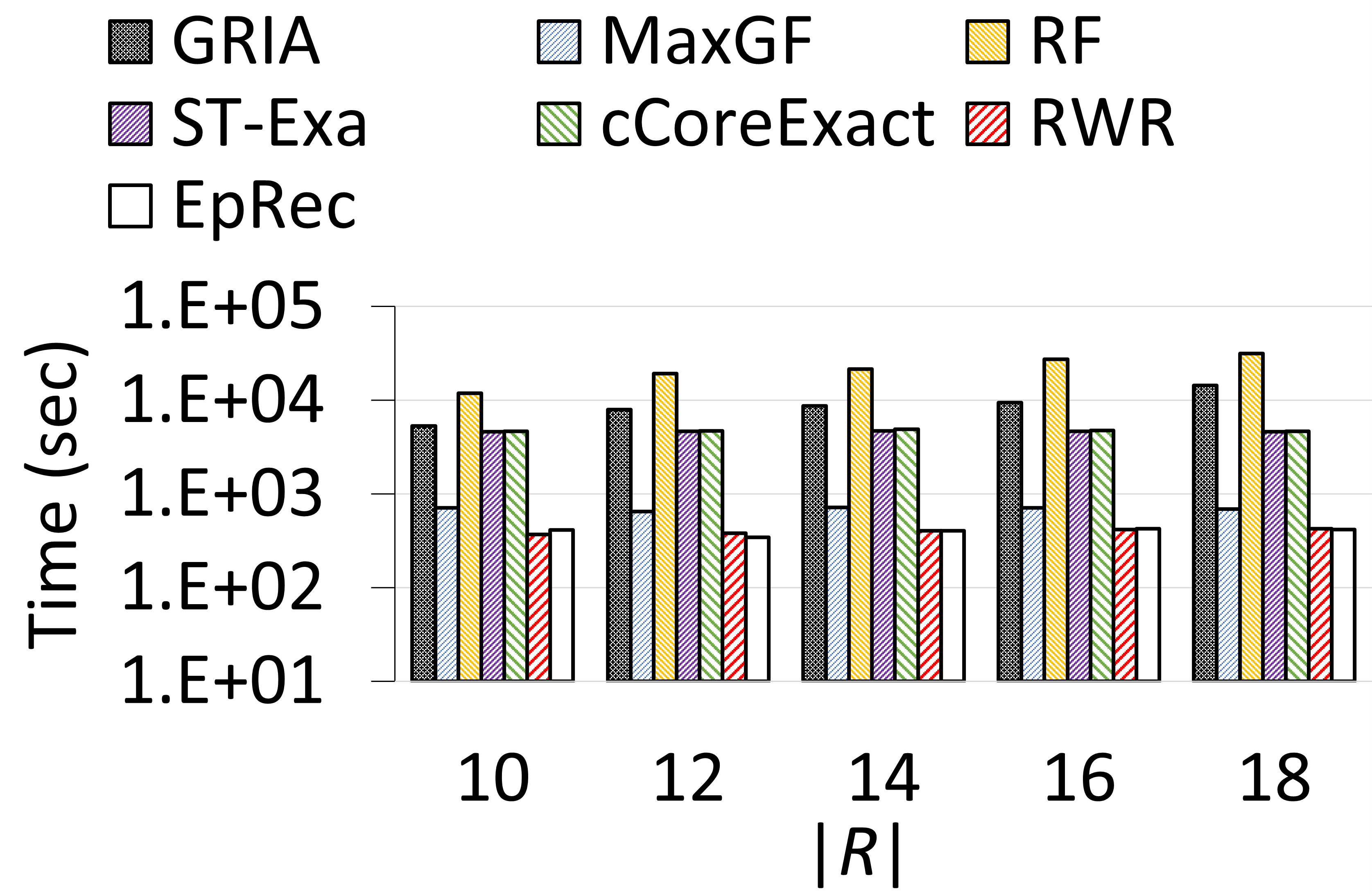}
		}
    \centering
	\subfigure[\emph{ca-GrQc}]{\label{GR_time}
	    \includegraphics[width=0.35\columnwidth]{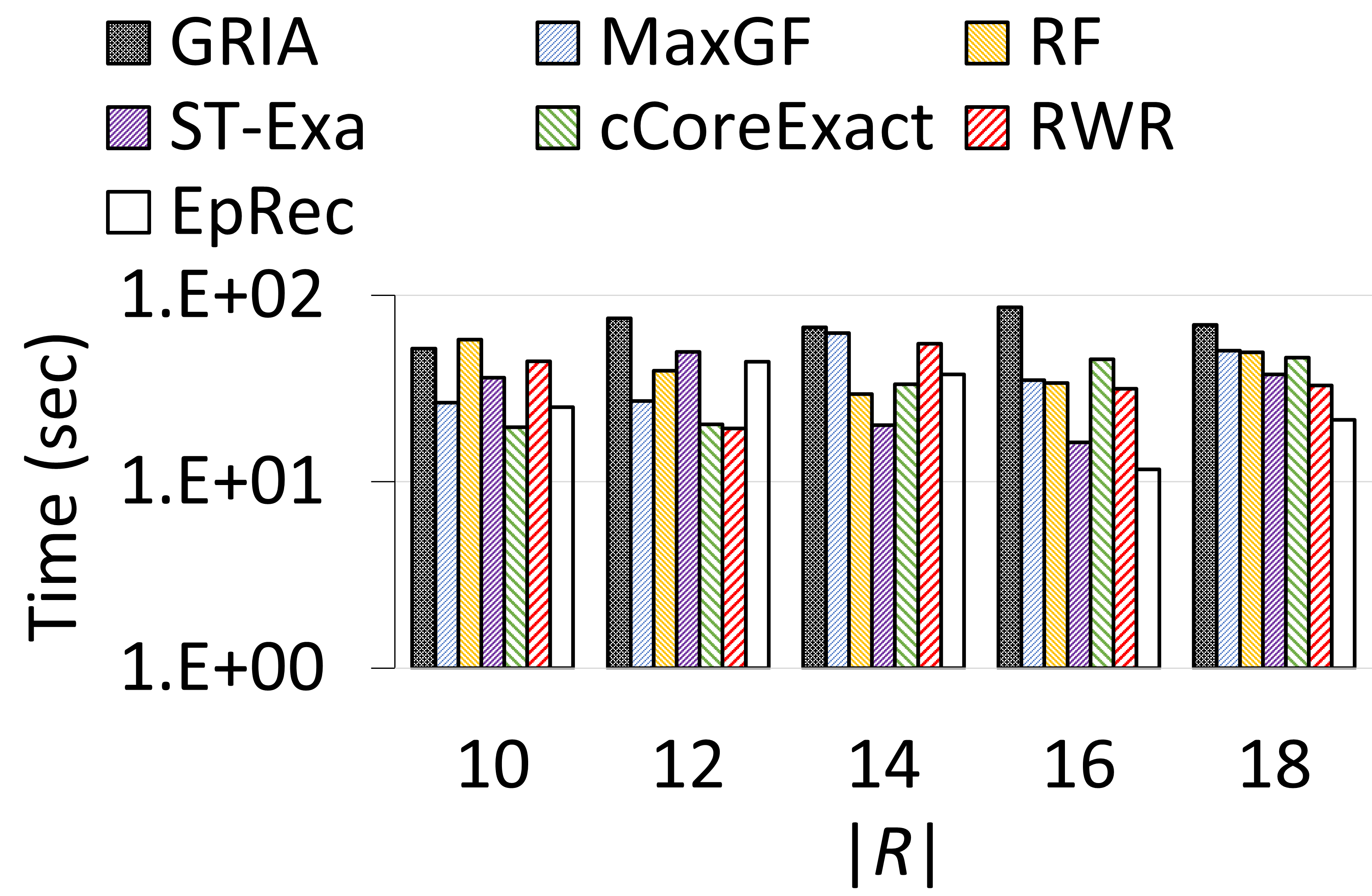}%
		}%
    \centering
	\subfigure[\emph{ca-HepPh}]{\label{HE_time}
	    \includegraphics[width=0.35\columnwidth]{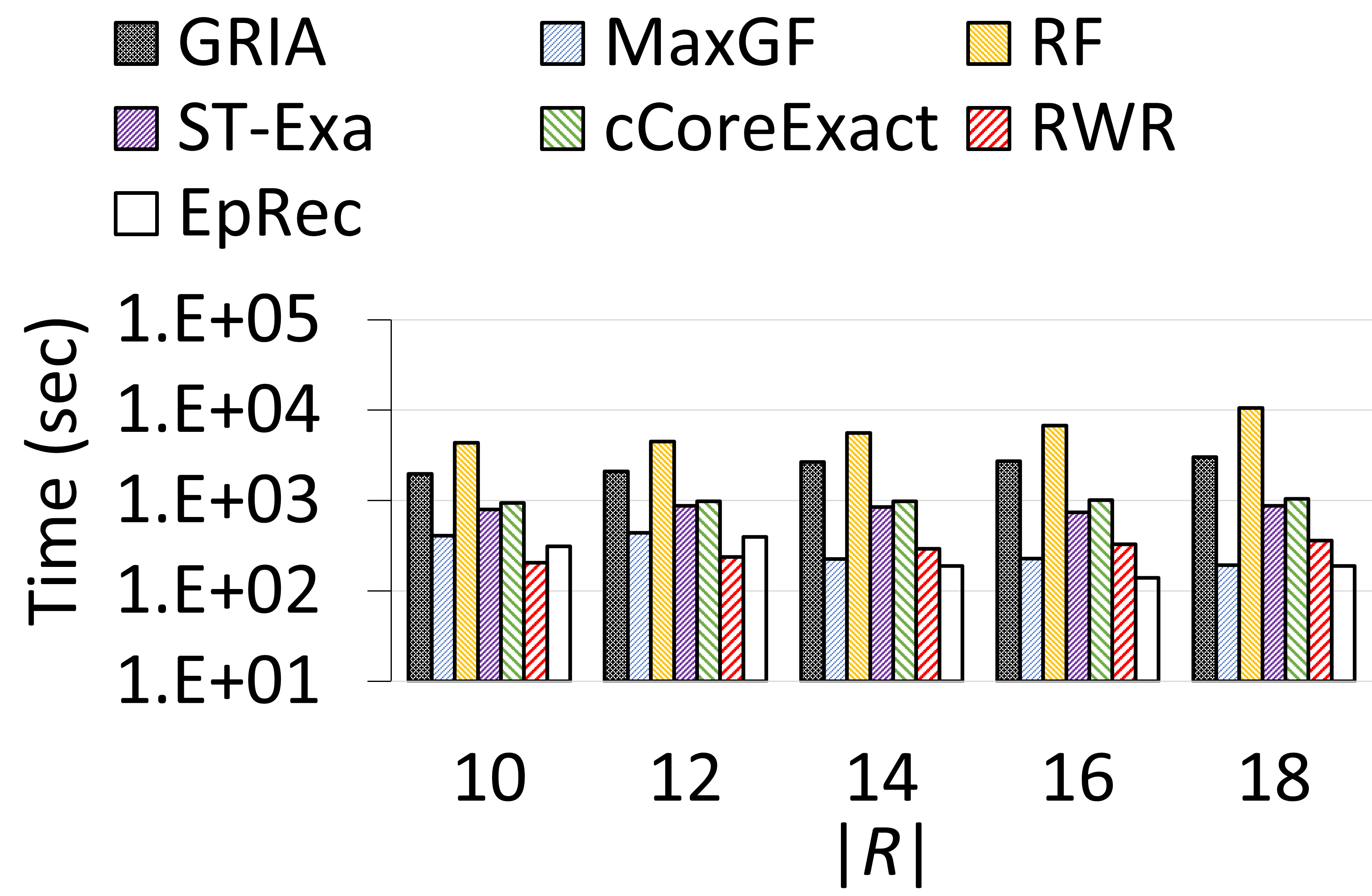}
		}
    \caption{The computation time while varying $|R|$}
    \label{skill_time}
\end{figure}

\noindent \textbf{Performance evaluation.}
Fig.~\ref{skill_obj} reports the objective ratio of \texttt{GRIA} and the baselines for different sizes of the required skill set $R$. 
Across all four datasets, \texttt{GRIA} consistently achieves the highest objective value and thus an objective ratio of $1$.  
Methods such as \texttt{MaxGF} and \texttt{ST-Exa} yield clearly lower ratios, while \texttt{RWR} typically performs better than these structure-only methods but still remains below \texttt{GRIA}.  
The trends in Fig.~\ref{MHA_skill_obj} are representative of those in Figures~\ref{VA_skill_obj}–\ref{HE_skill_obj}.  
As $|R|$ increases, the baselines generally improve because they tend to select larger onsite sets, which naturally increases the total collaboration score.  
Nevertheless, \texttt{GRIA} maintains a clear advantage by jointly accounting for skills, collaboration, and propagation risk.

Fig.~\ref{skill_time} compares the computation time of all algorithms as $|R|$ varies.  
\texttt{RWR} has the longest running time, followed by \texttt{GRIA}, while \texttt{ST-Exa} and \texttt{MaxGF} are considerably faster.  
Although \texttt{RWR} exhibits higher theoretical complexity, \texttt{MaxGF} attains the shortest runtime in practice, likely due to its use of neighbor pruning and core pruning to aggressively reduce the search space.  
\texttt{GRIA} is substantially faster than \texttt{RWR} but slower than the purely structural baselines, reflecting the additional cost of risk-aware construction and skill-preserving workforce refinement.  
%This trade-off between runtime and solution quality is consistent with the design goal of \texttt{GRIA}.

\begin{figure}[t]
    \centering
	\subfigure[\emph{ca-GrQc}]{\label{MHA_inf_val}
	    \includegraphics[width=0.35\columnwidth]{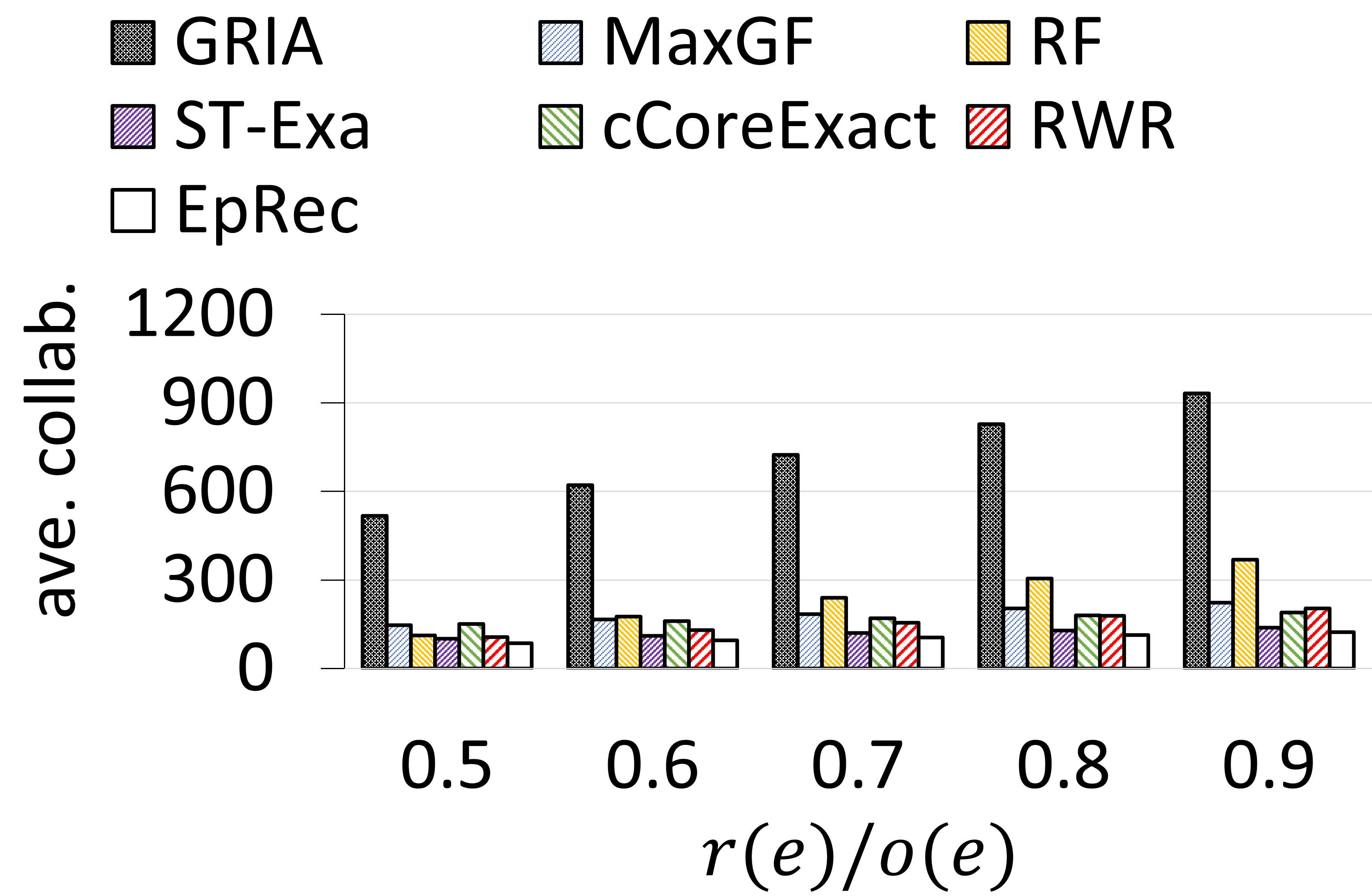}%
		}%
    \centering
	\subfigure[\emph{ca-HepPh}]{\label{VA_inf_val}
	    \includegraphics[width=0.35\columnwidth]{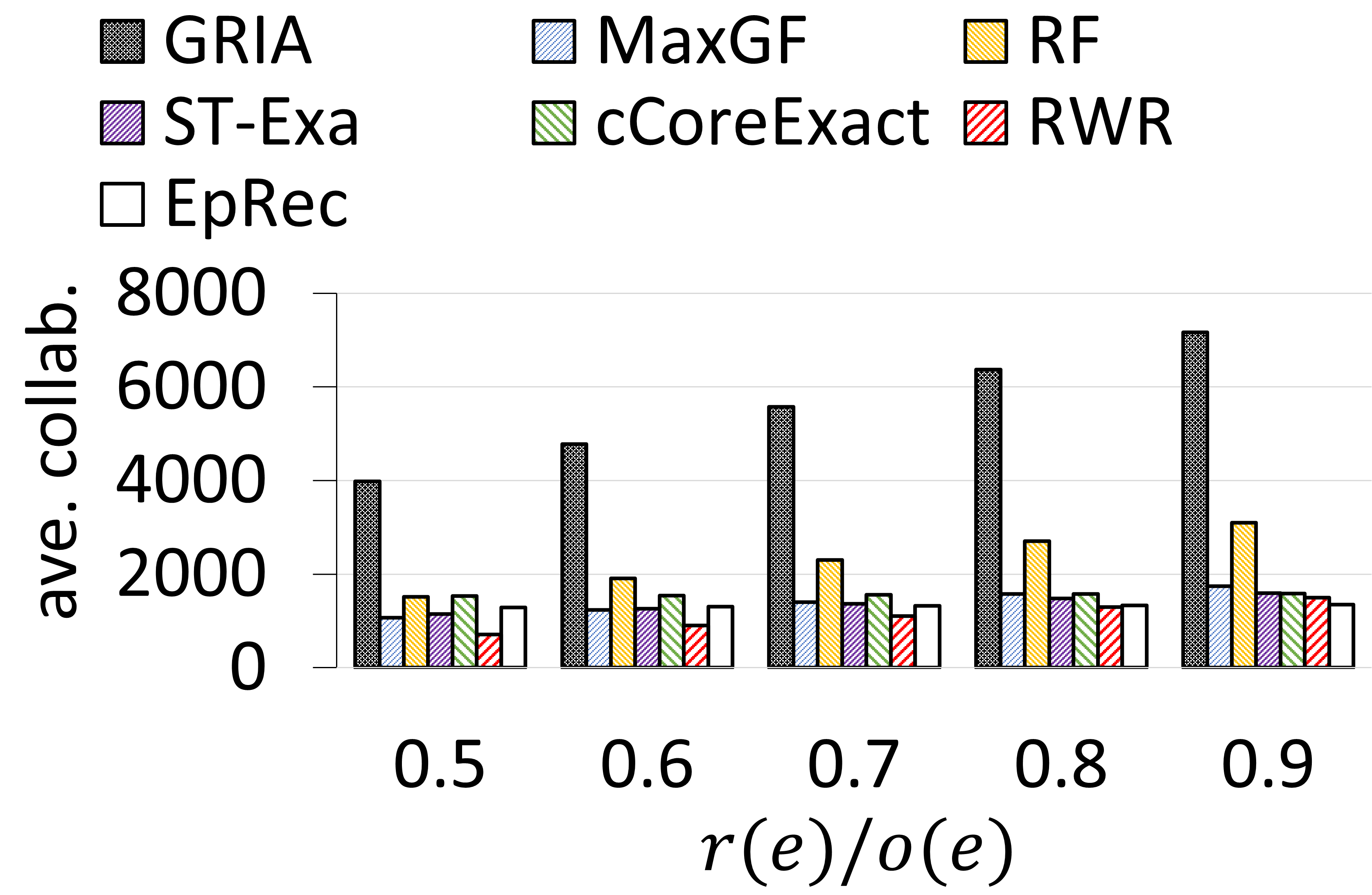}
		}
    \caption{The average collaboration while varying the remote/onsite effectiveness ratio}
    \label{fig:inf_val}
\end{figure}

\noindent\textbf{Impact of the remote/onsite effectiveness ratio.}
Fig.~\ref{fig:inf_val} evaluates the robustness of \texttt{GRIA} with respect to the relative effectiveness of remote versus onsite collaboration. We vary the ratio between the remote and onsite collaboration scores from $0.5$ to $0.9$ on \emph{ca-GrQc} and \emph{ca-HepPh}; the other two datasets, \emph{Manhattan} and \emph{Virginia}, exhibit similar trends and are therefore omitted for brevity. These ratios correspond to scenarios in which remote collaboration is between $50\%$ and $90\%$ as effective as onsite collaboration. Across all settings, \texttt{GRIA} consistently achieves higher objective values than all baselines, indicating that \texttt{GRIA} remains effective over a broad spectrum of hybrid-work configurations.

\section{Conclusion}
\label{sec:conclu}

In this paper, we formulate the RSHWC problem for hybrid workforce configuration on a two-layer social network with skill and epidemic-risk constraints, prove that RSHWC is NP-hard, and develop our \texttt{GRIA} algorithm that combines risk-aware workforce construction,  skill-preserving workforce refinement, and risk-reducing member replacement.  Experiments on four real-world datasets show that \texttt{GRIA} consistently achieves higher collaboration scores than state-of-the-art baselines under comparable risk budgets at moderate computational cost.

\linespread{0.95}
\bibliographystyle{splncs04}
\bibliography{reference}

\end{CJK}
\end{document}